\documentclass[letterpaper, 10 pt, conference]{ieeeconf}  

\IEEEoverridecommandlockouts                              
\usepackage[utf8]{inputenc}
\usepackage[T1]{fontenc}

\usepackage{cite}
\usepackage{amsmath,amssymb,amsfonts}
\usepackage{mathtools}
\usepackage{algorithmic}
\usepackage{graphicx}
\usepackage{textcomp}
\usepackage{cuted}
\def\BibTeX{{\rm B\kern-.05em{\sc i\kern-.025em b}\kern-.08em
    T\kern-.1667em\lower.7ex\hbox{E}\kern-.125emX}}

\usepackage{pgfplots}
\usepgfplotslibrary{statistics,groupplots}
\usetikzlibrary{spy}

\newtheorem{lemma}{Lemma}[section]
\newtheorem{theorem}{Theorem}[section]
\newtheorem{corollary}{Corollary}[section]

\newtheorem{remark}{Remark}[section]

\usepackage{cuted}

\title{\LARGE \bf An Extended Kuramoto Model for Frequency and Phase 
Synchronization in Delay-Free Networks with Finite Number of Agents*
}

\author{Andreas Bathelt$^{1}$, Vimukthi Herath$^{2}$, and Thomas 
Dallmann$^{1,2}$
\thanks{*This work has received funding by the German Federal Ministry of 
      Education and Research (BMBF) in the course of the 6GEM research hub 
      under grant number 16KISK038 and by the German Research Foundation (Deutsche Forschungsgemeinschaft, DFG) through project ”Coordinated multipoint operation for joint communication and radar sensing - JCRS CoMP” under project number 504990291.}
\thanks{$^{1}$Andreas Bathelt and Thomas Dallmann are with the Fraunhofer 
Institute for High Frequency 
            Physics and Radar Techniques FHR, Fraunhoferstraße 20, 53343 
            Wachtberg, 
            Germany (e-mail: andreas.bathelt@fhr.fraunhofer.de, 
            thomas.dallmann@fhr.fraunhofer.de)}%
\thanks{$^{2}$Vimukthi Herath and Thomas Dallmann are with the Radio Technologies for Automated and Connected Vehicles Research Group of Technische Universität Ilmenau, Ilmenau, Germany (email: vimukthi.herath@tu-ilmenau.de, thomas.dallmann@tu-ilmenau.de)}%
}

\begin{document}

\maketitle
\thispagestyle{empty}
\pagestyle{empty}

\begin{abstract}
  Due to its description of a synchronization between oscillators, the Kuramoto 
  model is an ideal choice for a synchronisation algorithm in networked 
  systems. This requires to achieve not only a frequency 
  synchronization but also a phase synchronization -- something the standard 
  Kuramoto model can not provide for a finite number of agents. In this case, a 
  remaining phase difference is necessary to offset differences of the natural 
  frequencies. Setting the Kuramoto model into the context of dynamic consensus 
  and making use of the $n$th order discrete average consensus algorithm, this 
  paper extends the standard Kuramoto model in such a way that frequency and 
  phase synchronization are separated. This in turn leads to an algorithm 
  achieve the required frequency and phase synchronization also for a finite 
  number of agents. Simulations show the viability 
  of this extended Kuramoto model.
\end{abstract}

\section{INTRODUCTION}

  The evolving field of Integrated Communications and Sensing (ICAS) promises 
  to merge mobile communications and environmental sensing based on radar 
  technology into one system. For interference-free message exchange,  a 
  synchronization of time offsets (TO) and carrier frequency offset (CFO) among 
  the devices of the wireless system is required\cite{nasir2016timing}.
  This 
  is achieved by master-slave approaches, e.g., the Schmidl \& Cox algorithm or 
  Zadoff-Chu sequences 
  \cite{schmidl1997robust,omri2019synchronization}, where the base stations 
  embed synchronization signals in the transmitted messages and 
  the user equipment performs a signal synchronization to these in-coming 
  signals. However, the 
  operation as radar sensor network places even higher demands on TO and CFO 
  synchronization as methods have to provide delay-free and highly 
  accurate timing information, i.e., a clock synchronization
  \cite{thomae2021joint,thomae2023distributed}.
  
  For such a synchronization, a multitude of approaches are available. 
  Atomic or GNSS-disciplined clocks are the most notable hardware-based 
  approaches. Network-based time synchronization protocols were also developed. 
  Beginning with Reference Broadcast Synchronization \cite{Elson.2002}, 
  Precision Time Protocol or White Rabbit 
  \cite{IEEEInstrumentationandMeasurementSociety.2008,Lipinski.2011} are 
  the standard master-slave-based approaches in this area. Another general 
  approach are multi-agent-based, i.e., reference-less, methods like 
  consensus-based time synchronization algorithms, e.g., 
  \cite{Tian.2015,Schenato.2011}. One of the earliest representative of this 
  multi-agent, reference-less idea is however 
  the Kuramoto model.

  This model has its roots in the observation of spontaneous synchronizations 
  in nature. For sets of coupled, nearly identical 
  oscillators, it can be observed that the coupling forcing them to operate in 
  unison. Examples stretch from brain rhythms to synchronous hand 
  clapping\cite{pikovsky2015dynamics}. For these phenomena, Kuramoto 
  proposed a mathematical model based on harmonic oscillators and weak coupling 
  driven by the oscillator phase differences. His main contribution lies in 
  the derivation of a steady-state solution for an infinite number of 
  oscillators, which exists for a sufficiently strong coupling factor 
  \cite{kuramoto1975self}. The Kuramoto model can therefore be used for a 
  frequency synchronization without a central coordination. It is already shown 
  that the Kuramoto model can be applied to the synchronization of pulse radars 
  and to CFO synchronization 
  \cite{dallmann2020mutual,dallmann2021sampling,dallmann2023mutual}. However, 
  the mathematical formulation of the Kuramoto model leads to a remaining 
  phase difference across the oscillators in the case of a finite number of 
  agents.
   
  With respect to the agreement among the agents, the Kuramoto model is 
  mentioned as a special case of static consensus in, e.g., 
  \cite{OlfatiSaber.2007,Moreau.2005}. Static consensus refers to algorithms, 
  which bring local and constant (static) quantities into agreement, see, e.g., 
  \cite{OlfatiSaber.2004,Blondel.2005,OlfatiSaber.2007}. 
  In addition to static consensus, there is also dynamic consensus, which 
  brings local functions of time into agreement, see, 
  e.g., \cite{Kia.2019,Spanos.2005, Zhu.2010}. The decision value is then
  represented by the algorithms' state variable(s). Similar to the 
  remaining phase difference of the  Kuramoto model, the error bound derived in 
  \cite{Kia.2019} for the (basic) dynamic consensus algorithm of 
  \cite{Spanos.2005} shows also a remaining difference across the agents' state 
  variables.
  
  Regarding its remaining phase error for a finite agent number and independent 
  of possible delays in the information exchange, the standard Kuramoto model 
  can not be applied to TO clock synchronization as this requires an agreement 
  of frequency \textit{and} phase for a finite agent number. It is thus 
  necessary to 
  extend the Kuramoto model to meet this 
  requirement. As this phase error is similar to the state 
  error of the basic dynamic consensus algorithm of \cite{Spanos.2005}, the 
  Kuramoto model can be seen as dynamic consensus. This in turn motivates 
  a combination of the Kuramoto model with algorithms of dynamic 
  consensus providing consensus without remaining state error.
  
  This paper sets therefore the Kuramoto model into the context of dynamic 
  consensus with undelayed information exchange. Using the 
  derivations of 
  \cite{Kia.2019}, an error bound for 
  the phases with respect to the consensus phase will be given. Assuming 
  nearly identical oscillators, the non-linearity of the Kuramoto model is 
  evaded by the small angle approximation. Based on the $n$th order 
  discrete average consensus (NODAC) algorithm of \cite{Zhu.2010}, an extended 
  Kuramoto model will be derived which facilitates frequency \textit{and} phase 
  agreement. This extended model provides the means for 
  TO synchronization while the bounds give a worst-case 
  estimate on the phase error of standard Kuramoto model.
  
  This paper is structured as follows. Section II reviews basics 
  of the Kuramoto model and dynamic consensus. Section III connects then the 
  Kuramoto model to dynamic consensus and Section IV presents the main result 
  -- the extended Kuramoto model. Simulation results are given in 
  Section V. Section VI briefly outlines an application of the 
  extended Kuramoto algorithm for ICAS networks. A 
  summary and an 
  outlook are 
  given in Section VII.
  
\section{PRELIMINARIES}
    This section provides a brief review of relevant definitions and 
    theoretical structures of Kuramoto model and consensus algorithms. First, 
    an overview of the notation is given.

    \begin{table*}[t]
    \setcounter{equation}{9}
    \begin{equation}\label{eqErrorBound}
      |e_i(t)| \leq \sqrt{\left(e^{-\hat{\lambda}_2(t-t_0)}\|\Pi x(t_0)\| + \frac{\sup_{t_0\leq\tau\leq t}\|\Pi\Dot{u}(t)\|}{\hat{\lambda}_2}\right)^2 
      + \left(\frac{1}{\sqrt{N}}\sum_{j=1}^Nx_j(t_0)-u_j(t_0)\right)^2}    
    \end{equation}
    \setcounter{equation}{0}
    \hrule
    \end{table*}
  
  \subsection{Notation}
    Variables like $\theta_i,\, \varphi_i  \in \mathbb{R}$ refer to the 
    respective agent marked by the index, here $i$. Variables without indices, 
    e.g., $\theta,\, \varphi \in \mathbb{R}^N$, refer to the aggregation of the 
    respective entities of the $N$ agents of the network into one vector, i.e., 
    $\theta = \begin{bmatrix}\theta_1 & \theta_2 & \cdots \theta_N 
    \end{bmatrix}^\mathrm{T}$. Finally, an over-line, e.g., $\overline{\varphi} 
    \in \mathbb{R}$,
    refers to the consensus value or agreement value of the respective entity. The unity matrix is given by $I_N \in \mathbb{R}^{N \times N}$ and $1_N = \begin{bmatrix}1 & 1 & \cdots 1 \end{bmatrix}^\mathrm{T} \in \mathbb{R^N} $ denotes the 1-vector of dimension $N$.
 
  \subsection{Kuramoto model}
    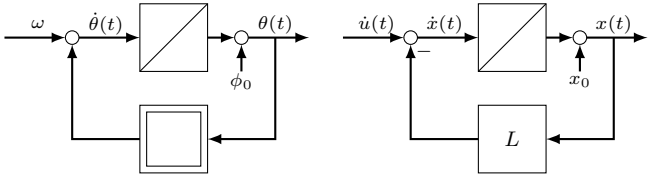
\begin{figure}[tb]
        \centering
        \begin{tikzpicture}[scale=0.9]

          
          \draw (0.5,3.2) node[text centered,text 
          width=0.5cm]{\scriptsize$\omega$};
          \draw[-latex,thick](0,3) -- (0.9,3);
 
          \draw(1,3) circle(0.1);

          \draw (1.5,3.2) node[text centered,text width=0.5cm]{\scriptsize$\Dot{\theta}(t)$};
          \draw[-latex,thick](1.1,3) -- (2,3);
        
          \draw(2,2.5) rectangle ++(1,1);
          \draw[](2,2.5) -- (3,3.5);

          \draw[-latex,thick](3,3) -- (3.4,3);
         
          \draw(3.5,3) circle(0.1);
          \draw[-latex,thick](3.5,2.5) -- (3.5,2.9);
          \draw (3.5,2.35) node[text centered,text 
          width=0.5cm]{\scriptsize$\phi_0$};

          \draw (4,3.2) node[text centered,text 
          width=0.5cm]{\scriptsize$\theta(t)$};
          \draw[-latex,thick](3.6,3) -- (4.5,3);
          \draw[-latex,thick](4,3) -- (4,1.5) -- (3,1.5);

          \draw(2,1) rectangle ++(1,1);
          \draw(2.1,1.1) rectangle ++(0.8,0.8);

          \draw[-latex,thick](2,1.5) -- (1,1.5) -- (1,2.9);


          \draw (5.5,3.2) node[text centered,text 
          width=0.5cm]{\scriptsize$\Dot{u}(t)$};
          \draw[-latex,thick](5,3) -- (5.9,3);

          \draw(6,3) circle(0.1);

          \draw (6.5,3.2) node[text centered,text 
          width=0.5cm]{\scriptsize$\Dot{x}(t)$};
          \draw[-latex,thick](6.1,3) -- (7,3);

          \draw(7,2.5) rectangle ++(1,1);
          \draw[](7,2.5) -- (8,3.5);

          \draw[-latex,thick](8,3) -- (8.4,3);

          \draw(8.5,3) circle(0.1);
          \draw[-latex,thick](8.5,2.5) -- (8.5,2.9);
          \draw (8.5,2.35) node[text centered,text 
          width=0.5cm]{\scriptsize$x_0$};

          \draw (9,3.2) node[text centered,text 
          width=0.5cm]{\scriptsize$x(t)$};https://www.overleaf.com/project/6525655574a4632ec92c974f
          \draw[-latex,thick](8.6,3) -- (9.5,3);
          \draw[-latex,thick](9,3) -- (9,1.5) -- (8,1.5);

          \draw(7,1) rectangle ++(1,1);
          \draw (7.5,1.5) node[text centered,text 
          width=0.5cm]{\footnotesize$L$};

          \draw[-latex,thick](7,1.5) -- (6,1.5) -- (6,2.9);
          \draw (6.2,2.8) node[text centered,text width=0.5cm]{\scriptsize$-$};
        \end{tikzpicture}
        \caption{Structures of Kuramoto model according to \eqref{eqKuramoto} 
        (left)  and of basic dynamic consensus algorithm according to 
        \eqref{eqBasicDynCons} \cite{Kia.2019} (right; for $\Dot{u}(t) \equiv 
        0$ 
        equivalent to static consensus)}
        \label{figKuramoto_DynConsensus}
    \end{figure}
    
    The phase $\varphi_i$ of a harmonic oscillator is given by  
    \begin{equation}\label{eqLocPhaseFct}
      \varphi_i(t) = \omega_it + \varphi_{0,i}\;,  
    \end{equation}
    where $\omega_i$ and $\varphi_{0,i}$ are the natural frequency and the 
    initial phase offset of oscillator $i$. The long term behavior of a 
    system consisting of loosely coupled oscillators with finite, nearly 
    identical cycles is discussed in \cite{Strogatz.2000}. 
    Accordingly, the rate of change of phase, $\Dot{\theta}_i(t)$, can be 
    expressed as given
    \begin{equation}\label{eqUniversal}
      \Dot{\theta}_i(t) = \omega_i + \sum_{j= 1,j \neq i}^N 
      T_{ij}\left(\theta_{\Delta,ij}(t)\right) \; ,
    \end{equation}
    where $T_{ij}(\cdot)$ is the interaction function between oscillators $i$ 
    and $j$, $\theta_{\Delta,ij}(t) = \theta_j(t)-\theta_i(t)$ is the phase 
    difference between oscillators $i$ 
    and $j$
    and $\theta_i(t_0) = \varphi_{0,i}$. For an all-to-all, equally weighed, 
    sinusoidal coupling, the interaction function $T_{ij}(\cdot)$ is 
    replaced with the $\sin$ function, and \eqref{eqUniversal} becomes
    \begin{equation}\label{eqKuramoto}
      \Dot{\theta}_i(t) = \omega_i + \frac{K}{N}\sum_{j=1,j \neq 
      i}^N\sin\left(\theta_{\Delta,ij}(t)\right)\; ,
    \end{equation}
    where $K$ is the coupling strength and again $\theta_{\Delta,ij}(t) = 
    \theta_j(t)-\theta_i(t)$. The 
    structure is shown on the left side of Fig. 
    \ref{figKuramoto_DynConsensus}. The combined 
    oscillation of individual oscillators results in a collective rhythm 
    given by the complex-valued order parameter 
    \begin{equation}\label{eqOrderParam}
        r(t)e^{i\psi(t)} = \frac{1}{N}\sum_{i= 1}^Ne^{\theta_i(t)}\;,
    \end{equation}
    through the phase coherence $r(t)$ and the average phase $\Psi(t)$. 

  \subsection{Consensus algorithms}
    \subsubsection{Network model}
      As given in \cite{OlfatiSaber.2004,Ren.2011}, the network is modelled by 
      a weighted, directed graph $G = (\mathcal{V},\mathcal{E},\mathcal{A})$, 
      where the nodes $\mathcal{V}=\{v_1, \cdots, v_k\}$ and the directed edges 
      $\mathcal{E}\subseteq \mathcal{V} \times \mathcal{V}$ represent the 
      agents and their communication connections. Whereas the orientation of an 
      edge $e_{ij} = (v_i,v_j)$ is from $v_i$ to $v_j$, the information flow 
      is in the reverse direction. The adjacency matrix $\mathcal{A}=[a_{ij}]$ 
      is induced by $\mathcal{E}$ as $e_{ij}\in \mathcal{E}\Leftrightarrow 
      a_{ij} > 0$\footnote{Time-varying weights or switching 
      topologies are not considered in this paper.}. The incidence matrix 
      $B\in\mathbb{R}^{N\times|\mathcal{E}|}$ is defined such that $b_{ij} = 1$ 
      if the edge $e_j$ ($j$ as a counting index of the edge, different from 
      $e_{ij}$) is incoming 
      to $v_i$, $b_{ij} = -1$ if $e_j$ is outgoing from $v_i$, and 
      $b_{ij} = 0$ else. Consensus within the network is reached, if every 
      other node is connected to at least one root node via a directed path 
      \cite{Blondel.2005,Ren.2011}. Average consensus
      (final decision value is mean of values of agents) is reached if each 
      agent has as many  neighbors as it is 
      neighbor to other agents (balanced graph), see \cite{OlfatiSaber.2004, 
      Kia.2019}.

    \subsubsection{Consensus protocols}
    Consensus algorithms are subdivided into static consensus, e.g., 
    \cite{OlfatiSaber.2004,OlfatiSaber.2007}, if the agreement is with respect 
    to a local constant, and dynamic consensus, e.g., 
    \cite{Kia.2019,Spanos.2005}, 
    if the agreement is with respect to a local input. Agreement is (in 
    principle) reached, if $x_i(t) = x_j(t)\; \forall i, j$ holds for all 
    states 
    variables $x_i, x_j$ of the agents $v_i$, $v_j$.  For a local agent 
    $v_i$, the basic protocol (feedback mechanism of the network) of static 
    consensus is given by
    \begin{equation}\label{eqBasicStatCons}
      \Dot{x}_i(t) = -\sum_{j= 1,j \neq i}^N a_{ij}\left(x_i(t) - x_j(t)\right)\; ,
    \end{equation}
    where $x_i(t)$ is the local state and $x_j(t)$ are the states of the 
    network's remaining agents. The network's state equation is
    \begin{equation}
      \Dot{x}(t) = - Lx(t) \; ,
    \end{equation}
    where $L=[l_{ij}]$ is the network Laplacian, defined by     
    \cite{OlfatiSaber.2004}
    \begin{equation}
       L = \Delta - \mathcal{A}\; , \quad 
       l_{ij} = \begin{cases}
       \sum_{j= 1,j \neq i}^N a_{ij} & i = j\\
       -a_{ij}                       & i \neq j
       \end{cases} \; ,
    \end{equation}
    where $\Delta$ is the out-degree matrix defined by 
    $\Delta_{ii} = \text{deg}_\text{out}(v_i)$.
    For dynamic consensus, the basic protocol of \cite{Spanos.2005} is given by
    \begin{equation}\label{eqBasicDynCons}
      \Dot{x}_i(t) = \Dot{u}_i(t) - \sum_{j= 1,j \neq i}^N a_{ij}\left(x_i(t) - 
      x_j(t)\right) \; ,
    \end{equation}
    whereas the feedback mechanism of the network is
    \begin{equation*}
      \Dot{x}(t) = \Dot{u}(t)- Lx(t) \; .
    \end{equation*}
    On the right side of Fig. \ref{figKuramoto_DynConsensus}, the respective 
    block diagram of this algorithm is shown.
     
    The explanation of average dynamic consensus of \cite{Kia.2019} also 
    includes the analysis of the error with respect to the network's agreement 
    function. This error is given by 
    \begin{equation}\label{eqErrorCons}
      e_i(t) = x_i(t) - \overline{u}(t),
    \end{equation}
    where $\overline{u}$ denotes the agreement function. The derived bound for 
    $e_i$ is shown by \eqref{eqErrorBound} (see top of the page). In 
    \eqref{eqErrorBound}, $t_0$ denotes the initial time, $\hat{\lambda}_2 = 
    \lambda_2\left(\frac{1}{2}(L+L^\mathrm{T})\right)$ describes a lower bound 
    on the convergence rate \cite{OlfatiSaber.2004}, and $\Pi = 
    I_N - \frac{1}{N}1_N1_N^\mathrm{T}$ is the orthogonal complement of the 
    agreement direction ($1_N$ for average consensus). For an
    initialization with $x_j(t_0)=u_j(t_0)$, the second term of 
    the square root vanishes.
    
    \setcounter{equation}{10}

    
    \subsubsection{NODAC ($n$th Order Discrete Average Consensus)}
      To overcome the issue of the remaining error (for a certain group of 
      input functions), Zhu and Martínez derived the NODAC algorithm in the 
      discrete-time setting; see \cite{Zhu.2010,Montijano.2014b}. For a fixed 
      network, this algorithm is given by\footnote{Note that the discrete-time 
      setting requires the weights $a_{ij}$ to fulfill certain conditions, see 
      \cite{Blondel.2005,Zhu.2010,OlfatiSaber.2007}. Since the paper focuses on 
      a continuous-time setting, these condition are not relevant in the 
      following and are hence here omitted.}
      \begin{subequations} \label{eqNODAC}
        \begin{align}
          \phantom{x_i^{[1]}(k+1)}
          \!\!\!&\begin{aligned}
            \mathllap{x_i^{[1]}(k+1)} &=  x_i^{[1]}(t)
                                 + \sum_{j = 1, \, i \neq j}^N 
                                 a_{ij}\left(x_j^{[1]}(k) - 
                                                          x_i^{[1]}(k)\right)\\
                                      &\qquad + \left(\Delta^{(n)}u_i\right)(k)
                                      \; ,
          \end{aligned}\!\!\!\label{eqDynConsensusa}\\
          \!\!\!&\begin{aligned}
             \mathllap{x_i^{[l]}(k+1)} &= x_i^{[l]}(k)
                                 + \sum_{j = 1,\, i \neq j}^N 
                                 a_{ij}\left(x_j^{[l]}(k) - 
                                                          x_i^{[l]}(k)\right)\\
                                      &\qquad + x_i^{[l-1]}(k+1) \; ,
          \end{aligned}\!\!\!
        \end{align}
      \end{subequations}
      where $x_i^{[l]},\, l \in \{1,\cdots,n\}$ denotes the state of 
      stage $l$ in agent $i$, $\Delta^{(n)}u_i$ the $n$-th order difference of 
      $u_i$, and $k \in \mathbb{N}$ the discrete time steps. For
      $m$-th order polynomials, a zero-error average consensus will be reached 
      for $n=m+1$. The idea of a stage-wise consensus on the respective 
      differences of the inputs can be carried over to the continuous-time 
      setting for the problem discussed in this paper.
      
\section{THE DYNAMIC CONSENSUS STRUCTURE OF THE KURAMOTO MODEL}
  In \cite{Moreau.2005,OlfatiSaber.2007}, the Kuramoto model \eqref{eqKuramoto} 
  is mentioned as an example for static consensus, with \cite{Moreau.2005} 
  making the restriction that all $\omega_i$ are the same. But, if they are not 
  the same, the structure of the Kuramoto model does no longer align with that 
  of a static consensus \eqref{eqBasicStatCons}. This section looks hence into 
  the consensus structure of the Kuramoto model and compares it with dynamic 
  consensus -- for an all-to-all (undirected) 
  and an arbitrary (directed) network, cf. \cite{Jadbabaie.2004}.
  
  \subsection{All-to-all network}
    Let the phase error be defined in similar to 
    \eqref{eqErrorCons} by 
    \begin{equation}\label{eqPhaseError}
      e_i(t) = \left(\theta_i(t) - \overline{\varphi}(t)\right)_{[-\pi,\pi]}\; ,
    \end{equation}
    where $\overline{\varphi}(t) = \overline{\omega}t + \overline{\varphi}$ 
    defines the consensus of the network. Also, due to the transformation 
    of the difference $\theta_j(t)-\theta_i(t)$ by the odd $\sin$ function in 
    \eqref{eqKuramoto} and the definition of the actual oscillator signal by 
    $\sin(\theta_i(t))$, it is not necessary to define the error as $e_i(t) \in 
    \mathbb{R}$. A phase error of $e_i(t) = \Tilde{e}_i(t) + 2N\pi,\, 
    N\in\mathbb{Z}$ is indistinguishable from $\Tilde{e}_i(t)$ in the final 
    sinusoidal signal. For this reason, the nominal error $\theta_i(t) - 
    \overline{\varphi}(t)$ is mapped to the interval $[-\pi,\pi]$ by 
    $(\cdot)_{[-\pi,\pi]} = \left((\cdot) + \pi\right)\mod 2\pi- \pi$. Although 
    the all-to-all network with a common weighing factor $\frac{K}{N}$ is by 
    definition balanced, the consensus phase function $\overline{\varphi}(t)$ 
    cannot be defined as the average of the local input function due to the 
    non-linearity of the $\sin$ function. However, using the order 
    parameter \eqref{eqOrderParam}, the consensus function 
    $\overline{\varphi}(t)$ can be retrieved after the transient behaviour of 
    the network since $\psi(t)$ of \eqref{eqOrderParam} behaves equal to the 
    consensus function such that 
    \begin{equation}
      \psi(t) = \overline{\varphi}(t) \; , \quad t \geq T_{tp} \; ,
    \end{equation}
    where $T$ is the duration of the transient period. Finally, the following 
    connection between the Kuramoto model and dynamic consensus can be made.
  
  \begin{lemma}\label{lemmaKuramotoDynCons}
    The Kuramoto model of \eqref{eqKuramoto} is a non-linear dynamic consensus with local phase functions $\varphi_i(t) = \omega_it + \varphi_{0,i}$ and $a_{ij}=\frac{K}{N}$. Furthermore, for $\theta_i(t_0) = \varphi_i(t_0)$ and the (final) mutual errors $\theta_j(t)-\theta_i(t)$ being small enough so that $\sin(\theta_j(t)-\theta_i(t))\approx\theta_j(t)-\theta_i(t)$, the remaining error of the agents in an all-to-all network configuration is bounded by
    \begin{equation}\label{eqBoundAlltoAll}
      \lim_{t\rightarrow\infty}|e_i(t)| \leq \frac{1}{\hat{\lambda}_2}\|\omega - 1_N\overline{\omega}\|\;,
    \end{equation}
    where $\overline{\omega} = \frac{1}{N}\sum_{i=1}^N \omega_i$.
  \end{lemma}
  
  \begin{proof}
    Setting $u_i(t) = \varphi_i(t)$ in \eqref{eqBasicDynCons} and taking into 
    consideration that $\frac{d}{dt}\varphi_i(t) = \omega_i$, the Kuramoto 
    model \eqref{eqKuramoto} describes the same dynamic consensus as 
    \eqref{eqBasicDynCons} -- except for the non-linearity within the sum of 
    the state (phase) differences due to the $\sin$ 
    function. By virtue of the small-angle approximation for small (final) 
    mutual errors, the linear protocol of \eqref{eqBasicDynCons} approximates 
    \eqref{eqKuramoto} sufficiently well so that the bound of the remaining 
    error \eqref{eqBoundAlltoAll} follows from the balanced graph of the 
    all-to-all case with the consensus direction (left eigenvalue of $L$ for 
    $\lambda = 0$) of $1_N$. Thus, \eqref{eqErrorBound} holds and 
    $\sup_{t_0\leq\tau\leq t}\|\Pi\Dot{\varphi}(\tau)\|$ in the first term of 
    \eqref{eqErrorBound} becomes
    \begin{equation*}
      \sup_{t_0\leq\tau\leq t}\|\Pi\Dot{\varphi}(\tau)\| = \|\Pi\omega\|  = \|\omega - 1_N\overline{\omega}\|\; ,
    \end{equation*}
    while the exponential becomes zero for $t\rightarrow\infty$ and the second 
    term vanishes based on $\theta_i(t_0) = \varphi_i(t_0)$.
  \end{proof}
  \begin{remark} \label{remarkConsDirection}
    In order to include the bound on the transient phase, i.e., 
    $e^{-\hat{\lambda}_2(t-t_0)}\|\Pi \theta(t_0)\|$, it would be necessary to 
    introduce an adjustment factor reflecting a bound on the worst-case 
    convergence due to the non-linear protocol \eqref{eqKuramoto}. The term 
    $e^{-\hat{\lambda}_2(t-t_0)}$ only reflects the network structure, but not 
    the non-linear behaviour of the differential equation.    
  \end{remark}
  \begin{remark}\label{remarkGammas}
    As shown by the decomposition of the consensus error in \cite{Kia.2019}, 
    the consensus frequency $\overline{\omega}$ and the (normed) left 
    eigenvalue $\gamma_L$ of $L$ for $\lambda = 0$ are connected through 
    \eqref{eqErrorCons} since, for $\gamma_L = \frac{1}{\sqrt{N}}1_N$, it 
    follows that
    \begin{align*}
      0 = \gamma_L^\mathrm{T}\Dot{e}(t) 
        &= \gamma_L^\mathrm{T}\left(-Lx(t) + \Dot{u}(t) - 1_N\Dot{\overline{u}}(t)\right)
    \end{align*}
    and hence $\Dot{\overline{u}}(t) = \frac{\gamma_L^\mathrm{T}\Dot{u}(t)}{\sum_{i=1}^N(\gamma_L)_i} = \frac{1}{N}1_N^\mathrm{T}\Dot{u}(t)$. For the Kuramoto model, this equation becomes (cf. \cite{Jadbabaie.2004})
    \begin{equation} \label{eqKuramotoConsDirection}
       0 = -\gamma^\mathrm{T}\left(\frac{K}{N}B\sin(B^\mathrm{T}\theta) + \omega - 1_N\Dot{\overline{\varphi}}\right) \; ,
    \end{equation}
    where the incidence matrix $B$ defines the Laplacian via $L = 
    BWB^\mathrm{T}$, where $W$ is the matrix of the edge weights -- for the 
    Kuramoto model, $W = \frac{K}{N}I$. Also, $\gamma$ marks the exact value as 
    induced by the network structure \textit{and} the non-linear protocol; 
    $\gamma_L$ follows through $L$ only from the network structure (of the 
    linear protocol). This shows that the consensus frequency converges to the 
    average of the individual frequencies only if 
    $-\gamma_L^\mathrm{T}(\frac{K}{N}B\sin(B^\mathrm{T}\theta) \approx 0$, 
    which is equivalent to 
    $\sin(\theta_j(t)-\theta_i(t))\approx\theta_j(t)-\theta_i(t)$.  In fact, 
    $\gamma$ is time-dependent due to the time-dependency induced by 
    $B\sin(B^\mathrm{T}\theta(t))$ (opposed to the time-invariance of $Lx(t)$), 
    i.e., $\gamma(t)$ is 
    defined such that
    \begin{equation}
       0 = \gamma^\mathrm{T}(t)\frac{K}{N}B\sin(B^\mathrm{T}\theta(t)) \; .
    \end{equation}
    For this reason, only the final mutual errors $\theta_j(t)-\theta_i(t)$ are 
    considered in the lemma above.
  \end{remark}
  
  \subsection{Arbitrary network}
    As analyzed in \cite{Jadbabaie.2004}, the Kuramoto model also holds for an 
    arbitrary (directed) network structure. For an arbitrary (non-balanced), 
    directed network, the consensus direction is no longer $1_N$, i.e., 
    $\gamma_L \neq \frac{1}{\sqrt{N}}1_N$, see \cite{OlfatiSaber.2004}. But, 
    similar to \eqref{eqKuramotoConsDirection}, 
    the left eigenvector/consensus direction still needs to fulfill
    \begin{equation}
      0 = -\gamma^\mathrm{T}(t)\left(\frac{K}{N}\Tilde{B}\sin(B^\mathrm{T}\theta(t)) + \omega - 1_N\Dot{\overline{\varphi}}\right)\; ,
    \end{equation}
    where $\Tilde{B}$ is defined such that $L = \Tilde{B}B^\mathrm{T} = \Delta 
    - \mathcal{A}$. The definition $L = BB^\mathrm{T}$ holds only if the graph 
    is undirected (e.g., an all-to-all network) and hence $L$ is symmetric. 
    That 
    is, $\Tilde{B} \in \mathbb{R}^{N\times|\mathcal{E}|}$ is defined as 
    $\Tilde{b}_{ij} = 1$ if the edge is incoming, $\Tilde{b}_{ij} = 0$ else. 
    Similar to Remark \ref{remarkGammas}, the network-based left 
    eigenvector $\gamma_L$, i.e., $\gamma_L^\mathrm{T}L = 0$, holds only in the 
    case of the small angle approximation. That is,  
    \begin{equation*}
      0 \approx 
      -\gamma_L^\mathrm{T}\frac{K}{N}\Tilde{B}\sin(B^\mathrm{T}\theta(t)) \;, 
      \quad t \geq T_{tp}
    \end{equation*}    
    holds true for the phase $\theta(t)$. Thus, 
    for an arbitrary, connected network the result as given below follows. 
       
    \begin{corollary}
      For an arbitrary, connected network structure, the Kuramoto model is a 
      non-linear dynamic consensus with local phase functions $\varphi_i(t) = 
      \omega_it + \varphi_{0,i}$ and $a_{ij}=\frac{K}{N}$. Furthermore, for 
      $\theta_i(t_0) = \varphi_i(t_0)$ and the (final) mutual errors 
      $\theta_j(t)-\theta_i(t)$ being small enough so that 
      $\sin(\theta_j(t)-\theta_i(t))\approx\theta_j(t)-\theta_i(t)$, the 
      remaining error of the agents is bounded by
      \begin{equation}\label{eqBoundArbitrary}
        \lim_{t\rightarrow\infty}|e_i(t)| \leq
        \frac{1}{\lambda_2}\|(I - \gamma_L\gamma_L^\mathrm{T})(\omega-1_N\overline{\omega})\| \; .
      \end{equation}
    \end{corollary}
  
    \begin{proof}
       The first part follows directly from Lemma \ref{lemmaKuramotoDynCons}. 
       Adjusted for $\gamma_L \neq \frac{1}{\sqrt{N}}1_N$, the error bound is 
       given by the derivation of \cite{Kia.2019}. The error $e(t)$ is 
       decomposed into its agreement and disagreement directions as
       \begin{equation*}
         T^\mathrm{T}e(t) =  
           \begin{bmatrix} {\Tilde{e}}_{agr}(t)\\ {\Tilde{e}}_{dis}(t) 
           \end{bmatrix}\, , \quad {\Tilde{e}}_{agr}(t) \in \mathbb{R}\, ,  
           \; {\Tilde{e}}_{dis}(t) \in \mathbb{R}^{N-1}
        \end{equation*}
        where $T = \begin{bmatrix} \gamma_L & R \end{bmatrix} \in \mathbb{R}^{N 
        \times N}$ with $R \in \mathbb{R}^{N \times (N-1)}$ and 
        $\|\gamma_L\|=1$ such that $TT^\mathrm{T}=T^\mathrm{T}T=I$.
        Assuming that the small angle approximation holds, the derivative of 
        the disagreement direction follows with 
        \begin{equation*}
        \Dot{\Tilde{e}}_{dis}(t) = -R^\mathrm{T}LR\Tilde{e}_{dis} 
        + R^\mathrm{T}(\Dot{\varphi}(t) - 1_N\Dot{\overline{\varphi}}(t))  \; .
        \end{equation*}
        Following the argument leading to the error bound of average 
        consensus presented in \cite{Kia.2019}, the orthogonal projection onto 
        the complement of the consensus direction now becomes $\Pi = I - 
        \gamma_L\gamma_L^\mathrm{T}$ and \eqref{eqBoundArbitrary} follows for 
        an arbitrarily connected Kuramoto model. Also, $\hat{\lambda}_2$ is 
        replaced by $\lambda_2 = \lambda_2(L)$ since $\hat{\lambda}_2$ can only 
        be used for the convergence speed of the disagreement direction in the 
        case of balanced networks but not in the case of arbitrary networks, 
        see \cite[Sec. VIII]{OlfatiSaber.2004}.
    \end{proof}

    \begin{remark}
      The bound \eqref{eqBoundArbitrary} can also be used to define a bound 
      based on \eqref{eqKuramotoConsDirection}. Assuming a 
      linearized version of $\eqref{eqKuramotoConsDirection}$, the same 
      derivation as done for \eqref{eqBoundArbitrary} leads to
      \begin{equation}
        \lim_{t\rightarrow\infty}|e_i(t)| \leq
        \frac{1}{\lambda_2}\|(I - \gamma\gamma^\mathrm{T})(\omega-1_N\overline{\omega})\|\;,          
      \end{equation}
      where $\gamma$ is that of \eqref{eqKuramotoConsDirection}.
    \end{remark}
  
\section{EXTENDED KURAMOTO MODEL}
     With the established connection between the Kuramoto model and dynamic 
     consensus, the NODAC algorithm is used to derive main result of the paper 
     -- a two-staged, extended version of the Kuramoto model yielding a 
     zero-phase error for a limited number of agents. In general, each of the 
     $n$ stages of the NODAC algorithm \eqref{eqNODAC} perform an individual 
     consensus on the $l$-th difference, $l = 0,\cdots, n-1$, of the input 
     function. Transferring this idea to continuous-time, it becomes the $l$-th 
     derivative. For the first order polynomial of the input function 
     $\varphi(t)$, the order of the NODAC becomes $n=2$ so that the extended 
     Kuramoto model resulting from fusing the continuous-time variant of a 2nd 
     order NODAC with the Kuramoto model is given by  
  \begin{subequations} \label{eqExtendedKuramoto}
    \begin{align}
  \Dot{\vartheta}_i(t) &= -\sum_{j= 1,j \neq i}^k 
  a_{ij}^\vartheta\left(\vartheta_{\Delta,ij}(t)\right)
  + \Ddot{\varphi}_i(t) \, , \! 
  \label{eqSynchConsensusFreq}\\
  \Dot{\theta}_i(t) &= \frac{K}{N}\sum_{j= 1,j \neq i}^k 
  a_{ij}^\theta\sin(\theta_{\Delta,ij}(t)) + \vartheta_i(t) 
  \, , \label{eqSynchConsensusAbs}
  \end{align}
  \end{subequations}
  where $\vartheta_{\Delta,ij}(t) = \vartheta_i(t) - \vartheta_j(t)$ and 
  $a_{ij}^\vartheta \in \mathbb{R}_0,\;a_{ij}^\theta \in \{0,1\}$ are 
  the weights of the frequency consensus and phase consensus represented by the 
  state variables $\vartheta_i(t)$ and $\theta_i(t)$. The trigonometric 
  function is only required in the second stage (phase consensus) since the 
  modulo $2\pi$ equivalence is only required for the phase.

    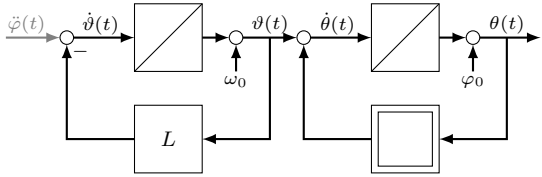
\begin{figure}[tb]
        \centering
        \begin{tikzpicture}[scale=0.9]


          \draw (0.4,3.2) node[text centered,text width=0.5cm,gray]{\scriptsize$\Ddot{\varphi}(t)$};
          \draw[-latex,thick,gray](0.1,3) -- (0.9,3);
 
          \draw(1,3) circle(0.1);

          \draw (1.5,3.2) node[text centered,text width=0.5cm]{\scriptsize$\Dot{\vartheta}(t)$};
          \draw[-latex,thick](1.1,3) -- (2,3);
        
          \draw(2,2.5) rectangle ++(1,1);
          \draw[](2,2.5) -- (3,3.5);

          \draw[-latex,thick](3,3) -- (3.4,3);

          \draw(3.5,3) circle(0.1);
          \draw[-latex,thick](3.5,2.5) -- (3.5,2.9);
          \draw (3.5,2.35) node[text centered,text width=0.5cm]{\scriptsize$\omega_0$};

          \draw (4,3.2) node[text centered,text width=0.5cm]{\scriptsize$\vartheta(t)$};
          \draw[-latex,thick](3.6,3) -- (4.4,3);
          \draw[-latex,thick](4,3) -- (4,1.5) -- (3,1.5);

          \draw(2,1) rectangle ++(1,1);
          \draw (2.5,1.5) node[text centered,text width=0.5cm]{\footnotesize$L$};

          \draw[-latex,thick](2,1.5) -- (1,1.5) -- (1,2.9);
          \draw (1.2,2.8) node[text centered,text width=0.5cm]{\scriptsize$-$};

 
          \draw(4.5,3) circle(0.1);

          \draw (5,3.2) node[text centered,text width=0.5cm]{\scriptsize$\Dot{\theta}(t)$};
          \draw[-latex,thick](4.6,3) -- (5.5,3);
        
          \draw(5.5,2.5) rectangle ++(1,1);
          \draw[](5.5,2.5) -- (6.5,3.5);

          \draw[-latex,thick](6.5,3) -- (6.9,3);

          \draw(7,3) circle(0.1);
          \draw[-latex,thick](7,2.5) -- (7,2.9);
          \draw (7,2.35) node[text centered,text width=0.5cm]{\scriptsize$\varphi_0$};

           \draw (7.5,3.2) node[text centered,text width=0.5cm]{\scriptsize$\theta(t)$};
          \draw[-latex,thick](7.1,3) -- (8,3);
 
          \draw[-latex,thick](7.5,3) -- (7.5,1.5) -- (6.5,1.5);

          \draw(5.5,1) rectangle ++(1,1);
          \draw(5.6,1.1) rectangle ++(0.8,0.8);

          \draw[-latex,thick](5.5,1.5) -- (4.5,1.5) -- (4.5,2.9);

       
        \end{tikzpicture}
        \caption{Structure of extended Kuramoto model according to \eqref{eqExtendedKuramoto}}
        \label{figExtendedKuramoto}
    \end{figure}
    The structure of this consensus is shown in Fig. \ref{figExtendedKuramoto}. 
    Given the local phase function \eqref{eqLocPhaseFct}, the network is 
    initialized with
    \begin{equation}\label{eqInitExtendedKuramoto}
       \vartheta_i(t_0) = \omega_i, \qquad 
       \theta_i(t_0)    = \varphi_i(t_0)\;.
    \end{equation}
    Separated from the second stage \eqref{eqSynchConsensusAbs}, the frequency 
    consensus of the first stage \eqref{eqSynchConsensusFreq} creates a common 
    reference frequency for the system. The second stage is then solely 
    responsible for bringing the phases into agreement. The 
    instantaneous frequency $\Dot{\theta}_i(t)$ brings only the phase 
    $\theta_i(t)$ into agreement since the additive term $\vartheta_i(t)$ 
    eventually converges to a common value. That is, $\Dot{\theta}_i(t)$ is 
    essentially only influenced by the sum term of \eqref{eqSynchConsensusAbs}. 
    The differences between the $\vartheta_i(t)$ of the transient period 
    create only temporary errors in the phase $\theta_i(t)$, which are then 
    compensated. The input $\Ddot{\varphi}(t)$ represents possible disturbances 
    in the local phase function \eqref{eqLocPhaseFct}, e.g., phase noise. Such 
    temporary or zero-mean disturbances are however compensated through an 
    adjustment of the frequency (and phase). Thus, the structure 
    of \eqref{eqExtendedKuramoto} recreates the situation of an equal frequency 
    for \eqref{eqSynchConsensusAbs} as assumed in \cite{Moreau.2005}. This 
    leads to the following result.

     \begin{theorem}
       Assuming an arbitrary, connected network and a local phase function 
       given by $\varphi_i(t) = \omega_it + \varphi_0$, then all local states 
       $\theta_i$ of the extended Kuramoto model \eqref{eqExtendedKuramoto} 
       converge to a common consensus function $\overline{\varphi}$ with no 
       remaining error. With an initialization of the network as in 
       \eqref{eqInitExtendedKuramoto}, the slope of $\overline{\varphi}$, i.e., 
       the consensus frequency, is given by 
       \begin{equation}\label{eqConsFreExtended}
          \lim_{t\rightarrow\infty}\vartheta_i(t) = \vartheta_\infty  = \frac{\gamma_L^\mathrm{T}\omega}{\sum_{i=1}^N(\gamma_L)_i} \; .
       \end{equation} 
     \end{theorem}
     \begin{proof}
       For an input function as given by \eqref{eqLocPhaseFct}, the second 
       derivative $\Ddot{\varphi}_i(t)$ is zero. Thus, 
       \eqref{eqSynchConsensusFreq} reduces to a static consensus whose 
       convergence is assured by the network being connected, cf. 
       \cite{OlfatiSaber.2004,OlfatiSaber.2007}. The consensus frequency in 
       \eqref{eqConsFreExtended} follows from Corollary 2 of 
       \cite{OlfatiSaber.2004}. Thus, \eqref{eqSynchConsensusAbs} becomes 
       essentially a Kuramoto model with all $\vartheta_i(t)$ converging to 
       $\vartheta_\infty$. Then, upon replacing
       $\omega$ with $1_N\vartheta_\infty$ and $\overline{\omega}$ with 
       $\frac{\gamma_L^\mathrm{T}1_N\vartheta_\infty}{\sum_{i=1}^N(\gamma_L)_i} 
       = \vartheta_\infty$, the error bounds 
       \eqref{eqBoundAlltoAll} and \eqref{eqBoundArbitrary} for the all-to-all 
       and the arbitrary network prove convergence of 
       the error \eqref{eqPhaseError} to zero and hence convergence of 
       \eqref{eqSynchConsensusAbs}.
        Furthermore, the error $\Dot{e}(t) =  
       \frac{K}{N}\Tilde{B}\sin(B^\mathrm{T}\theta(t)) + 1_N\vartheta_\infty - 
       1_N\overline{\vartheta}$ becomes $\Dot{e}(t) =  
       \frac{K}{N}\Tilde{B}\sin(B^\mathrm{T}\theta(t))$ since 
       $\overline{\vartheta} = 
       \frac{\gamma^\mathrm{T}1_N}{\sum_{i=1}^N(\gamma)_i} \vartheta_\infty = 
       \vartheta_\infty$ for any $\gamma$ (of \eqref{eqSynchConsensusAbs}). In 
       the steady state $\Dot{e}(t) = 0$, this means $\theta(t) = c1_N$.
     \end{proof}
 
\section{SIMULATIONS}
  This section provides two examples. The first example focuses on the bound 
  of the remaining error and the second example is a proof-of-concept of the 
  extended model. 
  
  \subsection{Simulation set-up}
    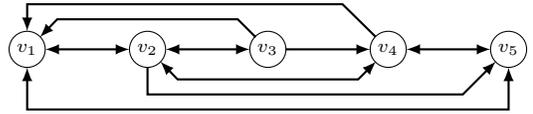
\begin{figure}[tb]
        \centering
        \begin{tikzpicture}[scale = 0.8]


          \draw(0,3) circle(0.3);
          \draw (0,3) node[text centered,text width=0.5cm]{\scriptsize$v_1$};

          \draw[latex-latex,thick](0.3,3) -- (1.7,3);
          \draw[latex-latex,thick](8,2.7) -- (8,2) -- (0,2) -- (0,2.7);



          \draw(2,3) circle(0.3);
          \draw (2,3) node[text centered,text width=0.5cm]{\scriptsize$v_2$};

          \draw[latex-latex,thick](2.3,3) -- (3.7,3);
          \draw[latex-latex,thick](5.7979,2.7979) -- (5.5,2.5) -- (2.5,2.5) -- (2.2121,2.7979);
          \draw[-latex,thick](2,2.7) -- (2,2.25) -- (7.24,2.25) -- 
                           (7.7879,2.7979);



          \draw(4,3) circle(0.3);
          \draw (4,3) node[text centered,text width=0.5cm]{\scriptsize$v_3$};

          \draw[-latex,thick](4.3,3) -- (5.7,3);
          \draw[-latex,thick](3.7979,3.2021) -- (3.5,3.5) -- (0.5,3.5) -- (0.2121,3.2021);  
          


          \draw(6,3) circle(0.3);
          \draw (6,3) node[text centered,text width=0.5cm]{\scriptsize$v_4$};

          \draw[latex-latex,thick](6.3,3) -- (7.7,3);
          \draw[-latex,thick] (5.7879,3.2121) -- (5.24,3.75) -- (0,3.75) -- (0,3.3); 



          \draw(8,3) circle(0.3);
          \draw (8,3) node[text centered,text width=0.5cm]{\scriptsize$v_5$};
        \end{tikzpicture}
        \caption{Schmeatic of network}
        \label{figNetwork}
    \end{figure}
    \begin{table}[tb]
        \caption{Parameters of the oscillators and agents}
        \label{tabParameters}
        \centering
        \begin{tabular}{l|ccccc}
          \hline
            agent  &    $v_1$&    $v_2$&    $v_3$&    $v_4$&    $v_5$\\
          \hline
            $\omega_i$ in rad/s
                   & $1.1$
                   & $0.8$
                   & $1$
                   & $1.3$
                   & $1.05$\\
            $\varphi_{0,i}$ in rad&  0.5& 2.5 & 1.5& 2 & 
            4.5\\
            neighbors $\mathcal{N}_i$ & \{2,5\} & \{1,3,4,5\} & \{1,2,4\} & \{1,2,5\} & \{1,4\} \\
          \hline
        \end{tabular}
    \end{table}

    Both simulations use the network and parameters as shown by Fig. 
    \ref{figNetwork} and Table \ref{tabParameters}. The network is an 
    arbitrary, directed network. The respective neighbors of the agents are also
    shown in Table \ref{tabParameters}. The weights of the edges and 
    $\frac{K}{N}$ are uniformly set to $1$. The frequency values are chosen to 
    ease the computational load and keep the simulation duration short -- the 
    pertinent effects are independent of the frequencies. The simulations run 
    in continuous-time with a maximum step-size of 0.01 sec. Since there are no 
    disturbances assumed in the phase functions \eqref{eqLocPhaseFct}, i.e., 
    $\Ddot{\varphi_i}(t) = 0, \; \forall t$, only their parameters are used for 
    the initialization, cf. \eqref{eqInitExtendedKuramoto}. That is, 
    the phase functions are not generated as explicit inputs of the 
    agents. Hence, they are also not shown in  Fig. 
    \ref{figNetwork}. Whereas $\varphi_{0,i}$ of agents 1 to 4 are chosen to 
    lie within the $[-\pi,\pi]$ bracket, $\varphi_{0,5}$ is set to be within 
    $[\pi,3\pi]$ in order to show that the phase functions $\theta_i(t)$ 
    converge to a common value up to an additional term of $2N\pi,\; N \in 
    \mathbb{N}$.       
    
  \subsection{Kuramoto model with error bound}
      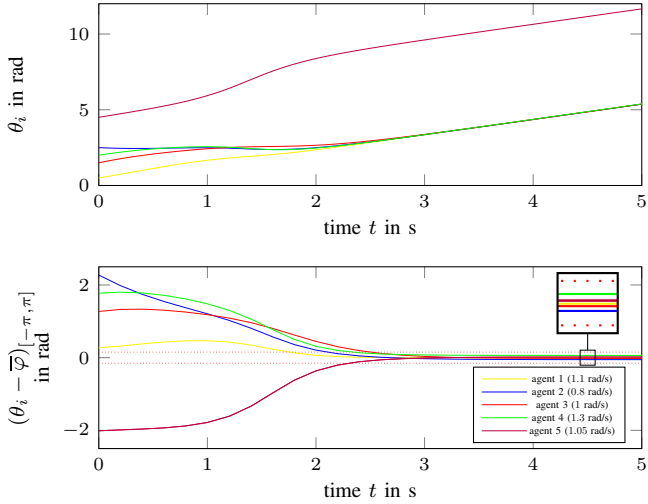
\begin{figure}
        \centering
        \pgfplotsset{every axis/.append style={
    label style={font=\footnotesize},
    tick label style={font=\scriptsize}
  },
  /pgf/number format/.cd,1000 sep={}
}

\begin{tikzpicture}[spy using outlines={rectangle, magnification=4,connect spies}]
  \begin{axis}[at={(0cm, 0cm)}, 
height=4cm, 
width=8.8cm,
xlabel=time $t$ in s,
x label style = {below=-2mm},
ylabel style={align=center,below=-1mm}, 
ylabel=$\theta_i$ in rad,
xmin = 0,
xmax = 5,
ymin = 0,
ymax = 12]%
\addplot[yellow, very thin] table[x index = 0, y index = 1] 
{Data/Phase_Kura1.txt};
\addplot[blue, very thin] table[x index = 0, y index = 2] 
{Data/Phase_Kura1.txt};
\addplot[red, very thin] table[x index =0, y index = 3] 
{Data/Phase_Kura1.txt};
\addplot[green, very thin] table[x index = 0, y index = 4] 
{Data/Phase_Kura1.txt};
\addplot[purple, very thin] table[x index = 0, y index = 5] 
{Data/Phase_Kura1.txt};
  \end{axis}

  \begin{axis}[at={(0cm, -3.5cm)}, 
height=4cm, 
width=8.8cm,
xlabel=time $t$ in s, 
x label style = {below=-2mm},
ylabel style={align=center,below=-1mm}, 
ylabel=$\left(\theta_i - \overline{\varphi}\right)_{[-\pi,\pi]}$\\in rad,
xmin = 0,
xmax = 5,
ymin = -2.5,
ymax = 2.5,
legend pos=south east,
legend style={nodes={scale=0.4, transform shape}}]%
\addplot[yellow, very thin] table[x index = 0, y index = 1] 
{Data/Err_Phase_Simple.txt};
\addplot[blue, very thin] table[x index = 0, y index = 2] 
{Data/Err_Phase_Simple.txt};
\addplot[red, very thin] table[x index =0, y index = 3] 
{Data/Err_Phase_Simple.txt};
\addplot[green, very thin] table[x index = 0, y index = 4] 
{Data/Err_Phase_Simple.txt};
\addplot[purple, very thin] table[x index = 0, y index = 5] 
{Data/Err_Phase_Simple.txt};
\addplot[purple, very thin] table[x index = 0, y index = 5] 
{Data/Err_Phase_Simple.txt};

\addplot[red, very thin, densely dotted] table[x index = 0, y index = 1] 
{Data/Bound_Simple.txt};
\addplot[red, very thin, densely dotted] table[x index = 0, y index = 2] 
{Data/Bound_Simple.txt};

\coordinate (spypoint) at (axis cs:4.5,0);
\coordinate (spyviewer) at (axis cs:4.5,1.5);
\spy[width=0.8cm,height=0.8cm] on (spypoint) in node [fill=white] at (spyviewer);

\addlegendentry{agent 1 ($1.1\text{ rad/s}$)}
\addlegendentry{agent 2 ($0.8\text{ rad/s}$)}
\addlegendentry{agent 3 ($1\text{ rad/s}$)}
\addlegendentry{agent 4 ($1.3\text{ rad/s}$)}
\addlegendentry{agent 5 ($1.05\text{ rad/s}$)}
\end{axis}
\end{tikzpicture}
        \vspace{-1.5\baselineskip}
        \caption{Simulation of Kuramoto model; phase $\theta_i$ (top) and phase error w.r.t. to consensus with error bound  (red, dotted lines) (bottom)}
        \label{figKuramotoSim}
    \end{figure}
    The results for the standard Kuramoto model \eqref{eqKuramoto} are shown 
    in Fig. \ref{figKuramotoSim} in terms of the phase functions $\theta_i(t)$ 
    and the error $\theta_i(t)-\overline{\varphi}(t)$ -- wrapped to 
    $[-\pi,\pi]$. The bound \eqref{eqBoundArbitrary}, calculated for $\gamma_L 
    = [0.6527,\, 0.2670,\, 0.0890,\, 0.3264,\, 0.6231]^\mathrm{T}$ and 
    $\lambda_2 = 2.382$, is also plotted. The magnification of the bottom 
    diagram of Fig. \ref{figKuramotoSim} shows that the errors of the 
    $\theta_i$ w.r.t. $\overline{\varphi}(t)$ converge to the interval defined 
    by the bounds but stay away from zero. Agent 5 shows the above-mentioned 
    effect of being pushed to the nearest $2N\pi$-equivalent of the actual 
    consensus function, which is given by $\overline{\varphi}(t) = 1.072 
    \text{ rad/s}\cdot t + 0.2281\text{ rad}$. The difference between 
    $\overline{\omega}$ calculated by 
    $\frac{\gamma_L^\mathrm{T}\omega}{\sum_{i=1}^N(\gamma_L)_i}$ and the actual 
    value due to $\psi(t)$ is $4.35\cdot 10^{-6}$. Since the largest    
    difference $\theta_i - \theta_j$ is $0.1172$ (agents 2 and 4) and 
    $\sin(\theta_i - \theta_j)=0.1169$, the small-angle 
    approximation applies. Hence, the application of the bound is also 
    theoretically reasonable. The bound is given by $0.1528$ whereas the 
    largest error w.r.t. $\overline{\varphi}(t)$ is $0.0627$. 
  
  \subsection{Extended Kuramoto model}
    \begin{figure}
        \centering
        \pgfplotsset{every axis/.append style={
    label style={font=\footnotesize},
    tick label style={font=\scriptsize}
  },
  /pgf/number format/.cd,1000 sep={}
}

\begin{tikzpicture}[spy using outlines={rectangle, magnification=4,connect spies}]
  \begin{axis}[at={(0cm, 0cm)}, 
height=4.0cm, 
width=4.6cm,
xlabel=time $t$ in s,
x label style = {below=-2mm},
ylabel style={align=center, below= -1mm}, 
ylabel=$\vartheta_i$ in rad/s,
xmin = 0,
xmax = 5,
ymin = 0.8,
ymax = 1.3,
legend pos=south east,
legend style={nodes={scale=0.425, transform shape}}]%
\addplot[yellow, very thin] table[x index = 0, y index = 1] 
{Data/Phase_Kura2.txt};
\addplot[blue, very thin] table[x index = 0, y index = 3] 
{Data/Phase_Kura2.txt};
\addplot[red, very thin] table[x index =0, y index = 5] 
{Data/Phase_Kura2.txt};
\addplot[green, very thin] table[x index = 0, y index = 7] 
{Data/Phase_Kura2.txt};
\addplot[purple, very thin] table[x index = 0, y index = 9] 
{Data/Phase_Kura2.txt};
  \end{axis}

  \begin{axis}[at={(4.2cm, 0cm)}, 
height=4.0cm, 
width=4.6cm,
xlabel=time $t$ in s,
x label style = {below=-2mm},
ylabel style={align=center, below= 3mm}, 
ylabel=$\theta_i$ in rad,
xmin = 0,
xmax = 5,
ymin = 0,
ymax = 12,
legend pos=south east,
legend style={nodes={scale=0.425, transform shape}}]%
\addplot[yellow, very thin] table[x index = 0, y index = 2] 
{Data/Phase_Kura2.txt};
\addplot[blue, very thin] table[x index = 0, y index = 4] 
{Data/Phase_Kura2.txt};
\addplot[red, very thin] table[x index =0, y index = 6] 
{Data/Phase_Kura2.txt};
\addplot[green, very thin] table[x index = 0, y index = 8] 
{Data/Phase_Kura2.txt};
\addplot[purple, very thin] table[x index = 0, y index = 10] 
{Data/Phase_Kura2.txt};
  \end{axis}

  \begin{axis}[at={(0cm, -3.5cm)}, 
height=4.0cm, 
width=8.8cm,
xlabel=time $t$ in s, 
x label style = {below=-2mm},
ylabel style={align=center, below= -1mm}, 
ylabel=$\left(\theta_i - \overline{\varphi}\right)_{[-\pi,\pi]}$\\in rad,
xmin = 0,
xmax = 5,
ymin = -2.5,
ymax = 2.5,
legend pos=south east,
legend style={nodes={scale=0.4, transform shape}}]%
\addplot[yellow,very thin] table[x index = 0, y index = 1] 
{Data/Err_Phase_Extended.txt};
\addplot[blue,very thin] table[x index = 0, y index = 2] 
{Data/Err_Phase_Extended.txt};
\addplot[red,very thin] table[x index =0, y index = 3] 
{Data/Err_Phase_Extended.txt};
\addplot[green,very thin] table[x index = 0, y index = 4] 
{Data/Err_Phase_Extended.txt};
\addplot[purple,very thin] table[x index = 0, y index = 5] 
{Data/Err_Phase_Extended.txt};

\coordinate (spypoint) at (axis cs:4.5,0);
\coordinate (spyviewer) at (axis cs:4.5,1.5);
\spy[width=0.8cm,height=0.8cm] on (spypoint) in node [fill=white] at (spyviewer);
\addlegendentry{agent 1 ($1.1\text{ rad/s}$)}
\addlegendentry{agent 2 ($0.8\text{ rad/s}$)}
\addlegendentry{agent 3 ($1\text{ rad/s}$)}
\addlegendentry{agent 4 ($1.3\text{ rad/s}$)}
\addlegendentry{agent 5 ($1.05\text{ rad/s}$)}
\end{axis}
\end{tikzpicture}
        \vspace{-1.5\baselineskip}
        \caption{Simulation of extended Kuramoto model; frequency $\vartheta_i$ (top left), phase $\theta_i$ (top right), and phase error w.r.t. to consensus (bottom)}
        \label{figExtendedKuramotoSim}
    \end{figure}
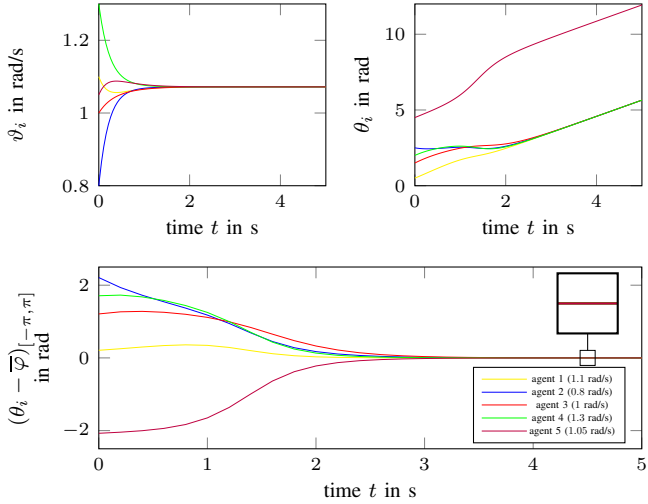

    The results for the extended Kuramoto model \eqref{eqExtendedKuramoto} are 
    shown in Fig. \ref{figExtendedKuramotoSim}. The results of the additional 
    frequency consensus stage, shown in the top left diagram, depict the 
    explicit agreement of the agents on $\overline{\omega} = 1.072 \text{ 
    rad/s}$ -- for the standard Kuramoto model this frequency follows 
    implicitly from the phase consensus. Similar to the standard Kuramoto, the 
    phases $\theta_i$ (top right diagram) show the same split between 
    $\overline{\varphi}(t)$ (for agent 1 to 4) and $\overline{\varphi}(t)+2\pi$ 
    (for agent 5). But, as shown by the bottom diagram, there now no remaining 
    error. Due to the explicit separation into frequency and phase consensus by 
    \eqref{eqSynchConsensusFreq} and \eqref{eqSynchConsensusAbs}, 
    $\overline{\varphi}_0$ also differs from the standard Kuramoto model and 
    the consensus phase is now given by $\overline{\varphi}(t) = 1.072 \text{ 
    rad/s}\cdot t + 0.2905\text{ rad}$.
      
    \begin{figure} 
        \centering
        \pgfplotsset{every axis/.append style={
    label style={font=\footnotesize},
    tick label style={font=\scriptsize}
  },
  /pgf/number format/.cd,1000 sep={}
}

\begin{tikzpicture}
  \begin{axis}[at={(0cm, 0cm)}, 
height=3.75cm, 
width=8.8cm,
xlabel=time $t$ in s,
x label style = {below=-2mm},
ylabel style={align=center, below=2mm}, 
ylabel=$\Dot{\theta}$ in rad/s,
xmin = 0,
xmax = 5,
ymin = -0.6,
ymax = 3.1,
legend pos=north east,
legend style={nodes={scale=0.4, transform shape}}]%
\addplot[yellow, very thin] table[x index = 0, y index = 1] 
{Data/Inst_phase_Kura2.txt};
\addplot[blue, very thin] table[x index = 0, y index = 2] 
{Data/Inst_phase_Kura2.txt};
\addplot[red, very thin] table[x index =0, y index = 3] 
{Data/Inst_phase_Kura2.txt};
\addplot[green, very thin] table[x index = 0, y index = 4] 
{Data/Inst_phase_Kura2.txt};
\addplot[purple, very thin] table[x index = 0, y index = 5] 
{Data/Inst_phase_Kura2.txt};
\addplot[yellow, very thin, dashed] table[x index = 0, y index = 1] 
{Data/Inst_phase_Kura1.txt};
\addplot[blue, very thin, dashed] table[x index = 0, y index = 2] 
{Data/Inst_phase_Kura1.txt};
\addplot[red, very thin, dashed] table[x index =0, y index = 3] 
{Data/Inst_phase_Kura1.txt};
\addplot[green, very thin, dashed] table[x index = 0, y index = 4] 
{Data/Inst_phase_Kura1.txt};
\addplot[purple, very thin, dashed] table[x index = 0, y index = 5] 
{Data/Inst_phase_Kura1.txt};
\addlegendentry{agent 1 ($1.1\text{ rad/s}$)}
\addlegendentry{agent 2 ($0.8\text{ rad/s}$)}
\addlegendentry{agent 3 ($1\text{ rad/s}$)}
\addlegendentry{agent 4 ($1.3\text{ rad/s}$)}
\addlegendentry{agent 5 ($1.05\text{ rad/s}$)}
  \end{axis}
\end{tikzpicture}
        \vspace{-1.5\baselineskip}
        \caption{Comparison of momentary frequency $\Dot{\theta}_i(t)$ between 
                 Kuramoto model (dashed) and extended Kuramoto model (solid)}
        \label{figComparisonMomFreq}
    \end{figure}
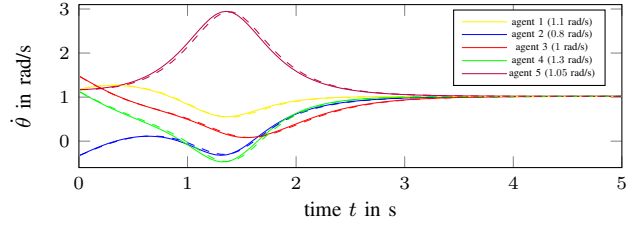

    This slightly different behaviour shows itself also in Fig. 
    \ref{figComparisonMomFreq}, which presents the instantaneous frequencies 
     $\Dot{\theta}_i$. Since these frequencies steer the phases 
     $\theta_i$ to the steady state, even the small difference 
     between zero-error and a maximum error of $0.0627$ creates dissenting 
     graphs of $\Dot{\theta}_i$ for the standard Kuramoto model and the 
     extended Kuramoto model. 
    
\section{APPLICATION TO ICAS NETWORKS}
  Synchronization in ICAS networks requires transmission of known 
  signals to identify CFO and TO for frequency and phase/time synchronization. 
  For this purpose, communication standards define different kinds of pilots 
  symbols or preambles.  While a discretized version of the extended Kuramoto 
  model could theoretically be applied to achieve frequency and phase 
  synchronization, the necessary sampling of a constant 
  pilot of frequency $\omega_p$ would require update rates of much larger than
  $\omega_i$. This in turn demands significant sampling and 
  clock rates from analog-to-digital converters and digital circuits. Moreover, 
  it is common that changes to pilot tones occur due to environmental 
  influences. 
 
  Thus, both models will be used, since the difference of the agreement 
  frequencies of both models is insignificant for practical applications. Let 
  each agent $v_i$  
  transmit a monofrequent 
  pilot tone at $\omega_i$ with repetition frequency $\Omega_{S,i}$ where 
  changes 
  to the tone only occur on tone-to-tone basis, i.e., the pilot tone does not 
  change its frequency during transmission \cite{dallmann2023mutual}. Also, let 
  delays be small enough to be neglectable. Using different aspects of the 
  pilot tones, CFO and TO synchronization is achieved. For CFO synchronization, 
  frequency agreement among the $\omega_i$ is given by $\Dot{\theta}_{i}(u)$ 
  of the discrete-time standard Kuramoto model
  \begin{equation}\label{eqCFOKuramoto}
      \Dot{\theta}_{i}(k) = \omega_i + \frac{K^\theta}{N}
     \sum_{j=1,j \neq i}^N\sin\theta_{\Delta,ij}(k) \; ,
  \end{equation}
  with an interval of $T_{S,i} = \frac{2\pi}{\Omega_{S,i}}$ for the time steps 
  $k\in \mathbb{N}$. Since the phases are not directly measurable,  
  $\theta_{\Delta,ij}(k)$ is defined by the total phase change due to each tone 
  during  $T_{S,i}$ as 
  \cite{dallmann2023mutual}
  \begin{equation} \label{eqnCFOphase}
     \theta_{\Delta,ij}(uk) = \theta_{\Delta,ij}(k-1) + \lambda 
                   T_{P,i}\hat{\omega}_{\Delta,ij}(k)\; , \quad k\geq 1\; ,
  \end{equation}
  where $\hat{\omega}_{\Delta,ij}(k)$ is the measured CFO between oscillators 
  $i$ and $j$ adjusted by $\theta_{i}$ and $\theta_{i}$, $T_{P,i} < T_{S,i}$ is 
  the pilot tone duration, and 
  $\lambda\leq 1$ is related to the sampling properties, see 
  \cite{dallmann2023mutual}. Thus, $T_{P,i}\hat{\omega}_{\Delta,ij}(k)$ is the 
  change of the phase difference over $T_{S,i}$. The update of 
  ${\theta}_{i}(k)$ is $\theta_{i}(k+1) = \theta_{i}(k) + 
  T_{S,i}\Dot{\theta}_{i}(k)$, cf. \cite{OlfatiSaber.2007}. For TO 
  synchronization, the rising edge of all tones must agree, i.e., the 
  differences between the instantaneous phases of the repetition frequencies 
  $\Omega_{S,i}$ are zero. This is given by $\Theta_{i}(k)$ of the 
  discrete-time extended Kuramoto 
  model (update interval again $T_{S,i} = \frac{1}{\Omega_i}$)
  \begin{subequations}
    \label{eqTOKuramoto}
    \begin{align}
      \Dot{\Omega}_{i}(k) &= -\sum_{j= 1,j \neq i}^k a_{ij}^\Omega 
       \Omega_{\Delta,ij}(k) \; , \\
      \Dot{\Theta}_{i}(k) &= \Omega_{i}(k) + \frac{K^\Theta}{N}
        \sum_{j= 1,j \neq  i}^k a_{ij}^\Theta \sin\Theta_{\Delta,ij}(k) \; .
    \end{align}
  \end{subequations}
  Since, again, a direct measurement of the frequencies and phases is 
  technically hardly feasible, the differences
  $\Omega_{\Delta,ij}(k)$ and 
  $\Theta_{\Delta,ij}(k)$ are estimated by \cite{dallmann2021sampling}
  \begin{subequations}
  \begin{align}
    \Omega_{\Delta,ij}(k) &\approx 
    \frac{\Omega_{i}(k)}{2\pi}
    \left(\Theta_{\Delta,ij}(k)-\Theta_{\Delta,ij}(k-1)\right)\; , \\
   \Theta_{\Delta,ij}(u) &= -\Lambda\left(\Omega_{i}(k)\hat{T}_{\Delta,ij}(k) + 
    2\pi P_{\Delta}(k)\right),
    \end{align}
    \end{subequations}
  where $\hat{T}_{\Delta,ij}(k)$ is the measured TO between the rising edges 
  of tones $i$ and $j$. Again, $\Lambda\leq 1$ relates to sampling properties, 
  see \cite{dallmann2021sampling}. Since tones are transmitted at 
  different repetition frequencies, the difference in tones transmitted by $i$ 
  and $j$ must be kept track of with $P_{\Delta}(u)\in\mathbb{N}$. 
  Both $\Omega_{i}(u)$ and $\Theta_{i}(k)$ are updated as $\Omega_{i}(k+1) = 
  \Omega_{i}(k) + T_{S,i}\Dot{\Omega}_{i}(k)$ and $\Theta_{i}(k+1) = 
  \Theta_{i}(k) + T_{S,i}\Dot{\Theta}_{i}(k)$, cf. \cite{OlfatiSaber.2007}. 

\section{CONCLUSION}
  This paper discussed the Kuramoto model in the context of dynamic 
  consensus. Based on the similarities between the dynamic consensus 
  and the Kuramoto model, bounds for the phase errors for 
  all-to-all and arbitrary networks are given. Also, using the idea 
  of a stage-wise consensus of the NODAC algorithm, an extended Kuramoto model 
  is derived, which yield a zero phase error also for a finite number of 
  agents. Future work will concentrate on the improvement of the bound 
  based on \cite{Jadbabaie.2004} to also capture the transient behavior, 
  delays, as well as an adjustment to allow for disturbances in the phase 
  functions, e.g., oscillator drift, by means of higher orders of the NODAC 
  algorithm. In addition, the focus will be on a further development of the 
  practical version as well as a hybrid approach dividing two stages into an 
  analog and a digital part.
  
\section*{APPENDIX} 
  This appendix gives a detailed description of the derivation of 
  \eqref{eqBoundArbitrary}. In \cite{Kia.2019}, the error bound 
  \eqref{eqErrorBound} is calculated for balanced networks, i.e., average 
  consensus, and hence for a specific value of the left eigenvector 
$\gamma_L$, 
  namely $\gamma_L = \frac{1}{\sqrt{N}}1_N$. For an arbitrary network this 
  derivation differs. 

  The starting point is again \eqref{eqErrorCons} or its network version
  \begin{equation*}
     e(t) = x(t) - 1_N\overline{u}(t)\; ,
  \end{equation*}
   where $\overline{u}(t)$ is also an arbitrary consensus function, i.e., not 
  necessarily the average of all $u_i(t)$. Following the logic of 
  \cite{Kia.2019}, a transformation according to the agreement and disagreement
  direction is defined as 
  \begin{equation*}
    T = \begin{bmatrix}
          \gamma_L & R
        \end{bmatrix}
  \end{equation*}
  with $TT^\mathrm{T}=T^\mathrm{T}T=I$ and $\|\gamma_L\|=1$. The error is 
  now given by
  \begin{equation*}
     \Tilde{e}(t) = T^\mathrm{T}e(t) = T^\mathrm{T}(x(t) - 1_N\overline{u}(t))
  \end{equation*} 
  and its derivative by
  \begin{align*}
     \Dot{\Tilde{e}}(t) = T^\mathrm{T}\Dot{e}(t) 
        &= T^\mathrm{T}(\Dot{x}(t) - 1_N\Dot{\overline{u}}(t))\\
        &= T^\mathrm{T}(-Lx(t) + \Dot{u}(t) - 1_N\Dot{\overline{u}}(t)) \; .
  \end{align*} 
  This equation can be rewritten as
  \begin{align*}
     \Dot{\Tilde{e}}(t) 
        =& T^\mathrm{T}\Dot{e}(t)\\
        =& -T^\mathrm{T}LTT^\mathrm{T}x(t) + T^\mathrm{T}\Dot{u}(t) - 
            T^\mathrm{T}1_N\Dot{\overline{u}}(t)\\
        =& -T^\mathrm{T}LTT^\mathrm{T}(x(t)-1_N\overline{u}(t)) + 
            (-T^\mathrm{T}LTT^\mathrm{T}1_N\overline{u}(t)))\\ 
          &+ T^\mathrm{T}\Dot{u}(t) - T^\mathrm{T}1_N\Dot{\overline{u}}(t)\\
        =&-T^\mathrm{T}LT\Tilde{e}(t) + T^\mathrm{T}\Dot{u}(t) - 
           T^\mathrm{T}1_N\Dot{\overline{u}}(t) \; .
  \end{align*}
  The last step follows since by assumption of a connected network 
  $LTT^\mathrm{T}1_N = L1_N = 0$. Now, splitting the error into the agreement 
  and disagreement directions, one gets
  \begin{align*}
    \begin{bmatrix}
      \Dot{\Tilde{e}}_{agr}(t)\\
      \Dot{\Tilde{e}}_{dis}(t) 
    \end{bmatrix} &=
    -\begin{bmatrix}
      \gamma_L^\mathrm{T}L\gamma_L & \gamma_L^\mathrm{T}LR\\
      R^\mathrm{T}L\gamma_L & R^\mathrm{T}LR
    \end{bmatrix}
    \begin{bmatrix}
      \Tilde{e}_{agr}(t)\\
      \Tilde{e}_{dis}(t) 
    \end{bmatrix}
    +
    \begin{bmatrix}
      \gamma_L^\mathrm{T}\\
      R^\mathrm{T}
    \end{bmatrix}(\Dot{u}(t) - 1_N\Dot{\overline{u}}(t))\\
    &=
    -\begin{bmatrix}
      0 & 0\\
      R^\mathrm{T}L\gamma_L & R^\mathrm{T}LR
    \end{bmatrix}
    \begin{bmatrix}
      \Tilde{e}_{agr}(t)\\
      \Tilde{e}_{dis}(t) 
    \end{bmatrix}
    +
    \begin{bmatrix}
      \gamma_L^\mathrm{T}\\
      R^\mathrm{T}
    \end{bmatrix}(\Dot{u}(t) - 1_N\Dot{\overline{u}}(t))\; . 
  \end{align*}
  Thus, assuming that the initialization was given by $x_i(t_0) = u_i(t_0)$ 
  which yields $\Tilde{e}_{agr}(t) = 0$, the derivative of the error of the 
  disagreement direction is given by
  \begin{equation*}
     \Dot{\Tilde{e}}_{dis}(t) = -R^\mathrm{T}LR\Tilde{e}_{dis} 
        + R^\mathrm{T}(\Dot{u}(t) - 1_N\Dot{\overline{u}}(t)) \; .
  \end{equation*}
  For a balanced network this is equivalent to (13b) of \cite{Kia.2019} since 
  in this case $\gamma_L = \frac{1}{\sqrt{N}}1_N$ and hence $R^\mathrm{T}1_N = 
  0$. Also, for the agreement direction error, it shows again that 
  $\Dot{\Tilde{e}}_{agr}(t) = 0$ follows as $\gamma_L^\mathrm{T}\Dot{u}(t) - 
  \gamma_L^\mathrm{T}1_N\Dot{\overline{u}}(t) = 0$ by definition of 
  $\Dot{\overline{u}}(t)$. The remainder of the derivation for 
  $\eqref{eqBoundArbitrary}$ now follows based on the fact that the 
  transformation $T^\mathrm{T}$ can also be expressed as a projection onto 
  $\gamma_L$ and its orthogonal complement and the same argument as made in 
  \cite{Kia.2019}.
  
\bibliographystyle{IEEEtran}
\bibliography{IEEEabrv,main}

\end{document}